\documentclass[aip,jmp,preprint]{revtex4-1}
\usepackage{graphicx}
\usepackage{amsmath}
\usepackage{amssymb}
\usepackage{amsthm}
\usepackage[all]{xy}
\newtheorem{thm}{Theorem}

\newtheorem{lem}{Lemma}
\newtheorem{prop}[lem]{Proposition}
\newtheorem{defn}[lem]{Definition}

\newtheorem*{thm*}{Theorem}

\numberwithin{equation}{section}

\usepackage{tikz-cd}

\usetikzlibrary{arrows}

\newcommand{\inbrac}[1]{\left[ #1 \right]}
\newcommand{\Hol}{\mathbf{Hol}}

\begin{document}
\title{A Categorical Equivalence between Generalized Holonomy Maps on a Connected Manifold and Principal Connections on Bundles over that Manifold}
\author{Sarita Rosenstock}
\email{rosensts@uci.edu}
\author{James Owen Weatherall}
\email{weatherj@uci.edu}
\affiliation{Department of Logic and Philosophy of Science, University of California--Irvine, 3151 Social Science Plaza A, Irvine, CA 92697-5100}
\begin{abstract}
A classic result in the foundations of Yang-Mills theory, due to J. W. Barrett [``Holonomy and Path Structures in General Relativity and Yang-Mills Theory.'' \emph{Int. J. Th. Phys.} \textbf{30}(9), (1991)], establishes that given a ``generalized'' holonomy map from the space of piece-wise smooth, closed curves based at some point of a manifold to a Lie group, there exists a principal bundle with that group as structure group and a principal connection on that bundle such that the holonomy map corresponds to the holonomies of that connection.  Barrett also provided one sense in which this ``recovery theorem'' yields a unique bundle, up to isomorphism.  Here we show that something stronger is true: with an appropriate definition of isomorphism between generalized holonomy maps, there is an equivalence of categories between the category whose objects are generalized holonomy maps on a smooth, connected manifold and whose arrows are holonomy isomorphisms, and the category whose objects are principal connections on principal bundles over a smooth, connected manifold.  This result clarifies, and somewhat improves upon, the sense of ``unique recovery'' in Barrett's theorems; it also makes precise a sense in which there is no loss of structure involved in moving from a principal bundle formulation of Yang-Mills theory to a holonomy, or ``loop'', formulation.
\end{abstract}
\maketitle

\section{Introduction}\label{sec-introduction}

There are two mathematical formalisms that are widely used for Yang-Mills theory.  One, well-known among physicists since the publication of the so-called ``Wu-Yang dictionary'',\citep{Wu+Yang} is the formalism of principal connections on principal bundles.\citep{Kobayashi+Nomizu, Trautman, Palais, Bleecker, Gockeler+Schuker, Deligne+Freed, WeatherallYM}  On this approach, a ``gauge field'' or ``Yang-Mills potential'' is a principal connection on some principal bundle over a relativistic spacetime, perhaps represented, in a section-dependent way, as the pullback of a connection one-form along a local section of the bundle; and the ``field strength'' is the curvature of this connection, again represented relative to some section of the principal bundle.  The choice of a section relative to which one represents these fields on spacetime corresponds to choosing a ``gauge''.  A second approach, attractive because it appears to do away with gauge-dependent potentials, is the formalism of ``loops'' or holonomies.\citep{Mandelstam, Barrett, Loll, Gambini+Pullin}  Here one directly associates closed, piece-wise smooth curves on spacetime with elements of some group, representing features of propagation along such curves, such as the phase-shifts in interference experiments associated with those closed curves.\footnote{Yet another approach is to characterize gauge theories directly using a generalized notion of the parallel transport maps induced by principal bundles with connections. See \citet{Schreiber+Waldorf} and \cite{Dumitrescu}. }

Given a principal connection on a principal bundle, one can immediately calculate the holonomies of that connection, relative to some point in the total space.  Conversely, a pair of classic results, due to Barrett,\citep{Barrett}  show that there is a certain sense in which, given appropriate ``holonomy data'' on a manifold $M$, there always exists a principal bundle over $M$ and a principal connection on that bundle such that the holonomy data arises as the holonomies of that connection, and that this bundle is, in a sense to be explained, unique.\footnote{\citet{Caetano+Picken} prove an analogous result using a different definition of loop space, which they take to eliminate certain disadvantages of Barrett's construction. The results presented here do not depend on which definition is used.}  Here we show that something stronger is true.  Given an appropriate notion of isomorphism between assignments of holonomy data, Barrett's reconstruction theorem gives the action on objects of a functor realizing a categorical equivalence between a category whose objects consist in specifications of holonomy data and whose arrows are holonomy isomorphisms and a category whose objects are principal connections on principal bundles over connected manifolds and whose arrows are connection-preserving principal bundle isomorphisms.

More precisely, let $M$ be a smooth, connected, paracompact Hausdorff manifold.\footnote{In what follows, we limit attention to manifolds that are smooth, Hausdorff, and paracompact, and will no longer state these assumptions explicitly.}  We will use $\bullet$ to denote reparameterized composition of curves with compatible endpoints, so that given two curves $\gamma_1:[0,1]\rightarrow M$ and $\gamma_2:[0,1]\rightarrow M$ such that $\gamma_2(0)=\gamma_1(1)$, we produce a curve $\gamma_2\bullet\gamma_1:[0,1]\rightarrow M$.\footnote{Everything discussed here is independent of the particular choice of reparameterization of the composition. The standard method is to define $\gamma_2 \bullet \gamma_1 (t) =  \begin{cases}
       \gamma(2t)  &  \text{ for } t \leq \frac{1}{2} \\
       \gamma(2(t-\frac{1}{2}) &  \text{ for } t \geq \frac{1}{2}
     \end{cases}$}  Given a curve $\gamma:[0,1]\rightarrow M$, meanwhile, we will take $\gamma^{-1}:[0,1]\rightarrow M$ to be the curve whose image is the same as $\gamma$'s, but whose orientation is reversed.  We will say that two curves $\gamma_1,\gamma_2:[0,1]\rightarrow M$ are \emph{thinly equivalent}, written $\gamma_1\sim\gamma_2$, if they agree on both endpoints and there exists a homotopy $h$ of $\gamma_1^{-1}\bullet\gamma_2$ to the null curve $id_{\gamma(0)}:[0,1]\rightarrow\gamma_1(0)$ such that the image of $h$ is included in the image of $\gamma_1^{-1}\bullet\gamma_2$.  Now let $x$ be some point of $M$ and denote by $L_x$ the collection of piece-wise smooth curves $\gamma:[0,1]\rightarrow M$ satisfying $\gamma(0)=\gamma(1)=x$.  A \emph{smooth finite-dimensional family of loops} at $x$ is a map $\tilde{\psi}:U\rightarrow L_x$, where $U$ is an open subset of $\mathbb{R}^n$ for any $n$, which is smooth in the sense that the associated map $\psi:U\times [0,1]\rightarrow M$ defined by $(u,t)\mapsto \tilde{\psi}[u](t)$ is continuous and smooth on subintervals $U\times[i_k,i_{k+1}]$, where $i_0=0<  i_i < \ldots < i_m = 1$ for some finite $m$.

With this background, we define the holonomy data mentioned above.  Let $M$ be a smooth manifold, let $G$ be a Lie group, and let $x$ be some point in $M$.  Then a \emph{generalized holonomy map} on $M$ with reference point $x$ and structure group $G$ is a map $H:L_x\rightarrow G$ satisfying the following properties: (1) for any $\gamma,\gamma'\in L_x$, if $\gamma$ and $\gamma'$ are thinly equivalent, then $H(\gamma)=H(\gamma')$; (2) for any $\gamma,\gamma'\in L_x$, $H(\gamma\bullet\gamma')=H(\gamma)H(\gamma')$; and (3) for any smooth finite-dimensional family of loops $\tilde{\psi}:U\rightarrow L_x$, the composite map $H\circ\tilde{\psi}:U\rightarrow L_x\rightarrow G$ is smooth.  For present purposes, the specification of a manifold $M$ and a generalized holonomy map $H:L_x\rightarrow G$, for some Lie group $G$ and point $x\in M$, constitutes a full specification of holonomy data; we will call the pair $(M,H)$ a \emph{holonomy model} for Yang-Mills theory.  (Note that we say nothing, here, of the dynamical relationship between $H$ and any distribution of charged matter.)  Barrett's results can then be stated as follows.

\begin{thm*}[Barrett reconstruction theorem] Fix a connected manifold $M$, a Lie group $G$, and a point $x\in M$, and let $H: L_{x} \rightarrow G$ be a generalized holonomy map. Then there exists a principal bundle $G \rightarrow P \overset{\pi}{\rightarrow} M$, a connection $\Gamma$ on $P$, and an element $u \in \pi^{-1}[x]$ such that $H = H_{\Gamma, u}$, where $H_{\Gamma, u}:L_x\rightarrow G$ is the holonomy map through $u$ determined by $\Gamma$.\footnote{Note that the (generalized) holonomy maps $H_{\Gamma, u}$ determined in this way really do depend on the choice of $u\in P$, even for connected manifolds; changing base point, even within a fiber, yields a holonomy map that is conjugate in $G$ to the one we began with, so if $u_2=u_1 g$, then $H_{\Gamma,u_2}(\gamma)=g^{-1}H_{\Gamma,u_1}(\gamma)g$ for any $\gamma\in L_{\pi(u_1)}$.  In the sequel, we make precise a sense in which these are nonetheless isomorphic holonomy maps.}
\end{thm*}

\begin{thm*}[Barrett representation theorem] The assignment of $(P, \Gamma, u)$ to generalized holonomy maps given in the above theorem is a bijection up to vertical principal bundle isomorphisms that preserve both the connection $\Gamma$ and the base point $u$.
\end{thm*}

Barrett's reconstruction theorem effectively establishes that holonomy data is sufficient to reconstruct a model of Yang-Mills theory in the sense of a principal connection on a principal bundle; the representation theorem, meanwhile, gives one sense in which this reconstruction is unique. But one might hope for something more regarding the uniqueness of the reconstruction.  In particular, on Barrett's approach, everything is done relative to fixed points $x\in M$ at which the closed curves are based and $u\in \pi^{-1}[x]$ at which the holonomies are based; nothing is said about the relationship between holonomy models associated with different base points, even though the base points play no role in the physics of Yang-Mills theory.\footnote{For another version of this worry, used to question the significance of Barrett's result, see \citet{Healey}.}  Moreover, the form of Barrett's results is highly suggestive: it appears that the relationship between holonomy maps and principal connections, properly construed, should be functorial.  Establishing this stronger result is the goal of the present paper.

In particular, we prove the following.  Let $\mathbf{PC}$ be the category (actually, groupoid) of principal connections on principal bundles over connected manifolds, with connection-preserving principal bundle isomorphisms as arrows, and let $\mathbf{Hol}$ be the category (or rather, again, groupoid) of holonomy models (as defined above) on connected manifolds, with ``holonomy isomorphisms'', to be defined in section \ref{sec:isomorphism}, as arrows.

\setcounter{thm}{1}
\begin{thm}\label{thm:cat} $\mathbf{Hol}$ and $\mathbf{PC}$ are equivalent as categories, with an equivalence that preserves empirical content in the sense of preserving holonomy data.\end{thm}
Our proof of Theorem \ref{thm:cat} depends on the following result concerning the notion of holonomy isomorphism we will presently define. We take this result to be of some interest in its own right.\setcounter{thm}{0}
\begin{thm}\label{thm:PBiso} Let $G \rightarrow P \overset{\pi}{\rightarrow} M$ and $G' \rightarrow P' \overset{\pi'}{\rightarrow} M'$ be principal bundles with principal connections $\Gamma$ and $\Gamma'$ respectively, and suppose that $M$ and $M'$ are connected.  Suppose there are points $u\in P$ and $u'\in P'$ such that the holonomy maps based at $u$ and $u'$ are isomorphic.  Then there is a connection-preserving principal bundle isomorphism between $P$ and $P'$.
\end{thm}

Note that the choice of a category of holonomy maps is not entirely straightforward, as there are several candidate notions for arrows between holonomy models.  Below we will identify two possible categories---$\mathbf{Hol}$, relative to which the theorem is stated, and $\mathbf{Hol}^*$---differing in their arrows, and show that they are related by a quotient functor that does not split.  For present purposes, we are agnostic as to which category is the ``right'' one, but find it expedient to prove our main theorem with $\mathbf{Hol}$, and then infer the analogue result for $\mathbf{Hol}^*$ as an immediate corollary.

We believe these results, taken together, substantially clarify the role of base points in Barrett's construction, by showing (1) how various changes of base point may be understood to induce an isomorphism of holonomy data and (2) that holonomy models related by holonomy isomorphisms in this sense correspond to isomorphic principal connections.  As we will discuss in the final section, Theorem \ref{thm:cat} also provides one sense in which there is no ``loss of structure'' involved in moving from the principal bundle formalism to the loop formalism (or vice-versa), despite claims by some that the latter is more parsimonious.\citep{Healey}

In the next section, we will discuss the possible choices of a category of holonomy maps and define the notion of holonomy isomorphism needed for Theorems \ref{thm:PBiso} and \ref{thm:cat}.  In the following two sections, we will prove Theorems \ref{thm:PBiso} and \ref{thm:cat}.  We conclude with a discussion of the interpretation of the results here, especially in connection with the Baez-Dolan-Bartels classification of forgetful functors,\footnote{See the discussion at http://math.ucr.edu/home/baez/qg-spring2004/discussion.html.} and a remark about how this work relates to recent results of \citet{Schreiber+Waldorf}.

\section{Holonomy isomorphism}\label{sec:isomorphism}

Consider a connected manifold $M$, a Lie group $G$, and a generalized holonomy map $H:L_x\rightarrow G$ for some point $x\in M$, in the sense defined above.  We are interested in developing a precise sense in which two such maps might be ``isomorphic'', in the sense of encoding the same physically relevant structure---i.e., the same ``holonomy data''.  To this end, we take the physically relevant structure of a generalized holonomy map to consist in the group theoretic structure of the assignments of elements of $G$ to piece-wise smooth closed curves in $M$.  This suggests that there are several ways in which two generalized holonomy maps might be understood to encode the same structure.  For one, consider diffeomorphic manifolds $M$ and $M'$.  Clearly, if $\Psi:M\rightarrow M'$ is a diffeomorphism, we can understand the generalized holonomy map $H':L_{\Psi(x)}\rightarrow G$ defined by $H'(\gamma)=H(\Psi^{-1}\circ\gamma)$ to encode the same holonomy data as $H$.  So we should take $H:L_x\rightarrow G$ and $H':L_{\psi(x)}\rightarrow G$ to be isomorphic if they are related by a diffeomorphism in this way.  Likewise, if $\phi:G\rightarrow G'$ is a Lie group isomorphism, the generalized holonomy map $H'=\phi\circ H:L_x\rightarrow G'$ may be understood to encode the same holonomy data as $H$, so we should take $H:L_x\rightarrow G$ and $H':L_x\rightarrow G'$ to be isomorphic if they are related by a Lie group isomorphism in this way.

There is a third sense in which two generalized holonomy maps may be understood to encode the same holonomy data, though it is somewhat more subtle to state.  The idea is that, as noted above, although a generalized holonomy map is defined relative to some base point $x\in M$, this base point plays no role in the physics.  Thus, we would like to understand generalized holonomy maps associated with different base points to encode the same data.  We do this as follows. Let $H: L_x \rightarrow G$ be as above and consider another point $y \in M$.  Let $\alpha$ be a piece-wise smooth curve in $M$ from $y$ to $x$. For all $\gamma \in L_x$, define $H_{\alpha}(\alpha^{-1} \bullet \gamma \bullet \alpha) := H(\gamma)$.  To extend $H_{\alpha}$ to all of $L_y$, recall that thinly equivalent curves must have the same holonomies. Thus for any $\gamma' \in L_y$, since $\gamma' \sim \alpha^{-1} \bullet \alpha \bullet \gamma' \bullet \alpha^{-1} \bullet \alpha$, $H_{\alpha}(\gamma') = H_{\alpha}( \alpha^{-1} \bullet \alpha \bullet \gamma' \bullet \alpha^{-1} \bullet \alpha ) = H(\alpha \bullet \gamma' \bullet \alpha^{-1})$.  There is thus a natural sense in which, relative to $\alpha$, $H$ and $H_{\alpha}$ may be understood to encode the same holonomy data.  In other words, given generalized holonomy maps $H:L_x\rightarrow G$ and $H':L_y\rightarrow G$, we should take $H$ and $H'$ to be isomorphic if there exists some piece-wise smooth curve $\alpha:[0,1]\rightarrow M$, with $\alpha(0) = y$ and $\alpha(1)=x$, such that $H'=H_{\alpha}$.

In connection with first two senses of isomorphism between generalized holonomy maps, it is natural to associate the induced holonomy isomorphism with a given choice of diffeomorphism or Lie group isomorphism.  In the third case, it is tempting to do likewise: that is, to associate a map with the curve $\alpha$ relating the generalized holonomy maps.  We might then define a map begin holonomy maps $H:L_x\rightarrow G$ to $H':L_{x'}\rightarrow G'$ on manifolds $M$ and $M'$ as an ordered triple $(\Psi, \alpha, \phi)$ where  $\Psi : M \rightarrow M'$ is a diffeomorphism, $\phi: G \rightarrow G'$ is a Lie group isomorphism, and $\alpha: [0,1] \to M$ is a piece-wise smooth curve satisfying $\alpha(0) = \Psi^{-1}(x')$ and $\alpha(1) = x$, and $\gamma \in L_x$, $\phi \circ H(\gamma) = H' (\Psi\circ(\alpha^{-1} \bullet \gamma \bullet \alpha))$.

But there are several reasons why this definition would be unsatisfactory.  The first is simple: we are looking for a notion of ``holonomy isomorphism'', and maps defined as just described would not be isomorphisms, since they are generally not invertible.  The second problem is closely related. Recall that thinly equivalent curves always have the same holonomies, and hence transformations of holonomy maps by thinly equivalent curves do not induce meaningfully different holonomy transformations.  Thus one might prefer, rather than indexing holonomy isomorphisms by curves, to use thin-equivalence classes of curves.  This modification leads to what we will call \emph{holonomy isomorphism}$^*$s, which are equivalent classes of maps as just described whose curves $\alpha$ are all thinly equivalent---i.e., triples $(\Psi, [\alpha]_{\sim}, \phi)$, where $[\alpha]_{\sim}$ is an equivalence class of curves under thin equivalence (for further details, see definition \ref{iso} below, with $[\alpha]_{\sim}$ substituted for $\underline{\alpha}$ in the relevant places).

With this definition of holonomy isomorphism$^*$, we have a natural candidate for an identity map associated with any generalized holonomy map $H:L_x\rightarrow G$: namely, the holonomy isomorphism $1_{H}:=(id_M , [id_x]_{\sim} , id_G ):H\rightarrow H$.  We also can define the composition of holonomy isomorphisms $(\Psi, [\alpha]_{\sim}, \phi): H \rightarrow H'$ and $(\Psi', [\alpha']_{\sim}, \phi'): H' \rightarrow H''$, by $(\Psi', [\alpha']_{\sim}, \phi') \circ (\Psi, [\alpha]_{\sim}, \phi) := (\Psi' \circ \Psi, [\alpha \bullet (\Psi^{-1}\circ\alpha')]_{\sim}, \phi' \circ \phi)$, where $[\alpha \bullet (\Psi^{-1}\circ \alpha')]_{\sim}$ is the equivalence class of curves generated by $\alpha \bullet (\Psi^{-1}\circ\alpha')$ for \emph{any} curves $\alpha\in[\alpha]_{\sim}$ and $\alpha'\in[\alpha']_{\sim}$.  We can thus define a category $\mathbf{Hol}^*$ of holonomy maps and holonomy isomorphism$^*$s.  As we show in Prop. \ref{groupoid} below, holonomy isomorphism$^*$s are isomorphisms (and thus $\mathbf{Hol}^*$ is a groupoid), and so we have also solved the first problem mentioned above.

Arguably, however, holonomy isomorphism$^*$s still do not give us what we want.  The reason is that, given two generalized holonomy maps $H:L_x\rightarrow G$ and $H':L_y\rightarrow G$, there may exist distinct curves $\alpha,\beta:[0,1]\rightarrow M$, both satisfying $\alpha(0)=\beta(0)=y$ and $\alpha(1)=\beta(1)=x$, and both such that $H'=H_{\alpha}=H_{\beta}$.  To count these as distinct isomorphisms would be to assert that there is a substantive (or at least, salient) difference in the way $\alpha$ and $\beta$ take $H$ to $H'$.  But since the physics depends only on the assignments of group elements to closed curves, if $\alpha$ and $\beta$ both provide the same ``translation'' from the assignments made by $H$ to the assignments made by $H'$, then nothing in the physics turns on which translation one picks, and this should be reflected in how we differentiate isomorphisms---i.e., we should not make a distinction if there is no salient difference. We address this issue by saying that curves $\alpha,\beta:[0,1]\rightarrow M$ are equivalent (relative to $H$ and $H'$) if $H'=H_{\alpha}=H_{\beta}$.

The considerations just described are summed up in the following definition of holonomy isomorphism.


%
%

\begin{defn}[Holonomy isomorphism]\label{iso}
Let $H:L_x\rightarrow G$ and $H':L_{x'}\rightarrow G'$ be (generalized) holonomy maps on manifolds $M$ and $M'$.  A \emph{holonomy isomorphism} from $H$ to $H'$ is an ordered triple $(\Psi, \underline{\alpha}, \phi)$ where  $\Psi : M \rightarrow M'$ is a diffeomorphism, $\phi: G \rightarrow G'$ is a Lie group isomorphism, and $\underline{\alpha}$ is an equivalence class of piece-wise smooth curves $\alpha: [0, 1] \rightarrow M$ satisfying $\alpha(0) = \Psi^{-1}(x')$ and $\alpha(1) = x$, which are all such that for any $\gamma \in L_x$, $\phi \circ H(\gamma) = H' (\Psi\circ(\alpha^{-1} \bullet \gamma \bullet \alpha))$. In other words, the following diagram commutes:
\begin{center}\leavevmode\xymatrix{
L_x \ar[rr]^{\overline{\alpha}} \ar[dd]_{H} & & L_{\Psi^{-1}(x')} \ar[rr]^{\psi} && L_{x'} \ar[dd]^{H'} \\
\\
H[L_x] \ar[rrrr]_{\phi} & & & & H[L_{x'}]}
\end{center}
Where $\psi : L_{\Psi^{-1}(x')} \rightarrow L_{x'}$ is defined by $\gamma \mapsto \Psi \circ \gamma$ and $\overline{\alpha} : L_x \rightarrow L_{\Psi^{-1}(x')} $ is defined by $\gamma \mapsto \alpha^{-1} \bullet \gamma \bullet \alpha$ for some element $\alpha$ of the equivalence class $\underline{\alpha}$.
\end{defn}

As with holonomy$^*$s, we can immediately define notions of identity map, composition, and inverse for holonomy isomorphisms, and thus define a category $\mathbf{Hol}$ of holonomy maps and holonomy models.  As with holonomy isomorphism$^*$s, we find that holonomy isomorphisms are, indeed, isomorphisms, as shown in the following proposition.

\begin{prop}\label{groupoid} $\mathbf{Hol}$ and $\mathbf{Hol}^*$ are groupoids.
\end{prop}
\begin{proof} The arguments in both cases are identical, and so we will consider just $\mathbf{Hol}$.  It is clear from the forgoing that (a) we have identity arrows for each object and (b) the composition of any two holonomy isomorphisms with appropriate domain and codomains yields a new holonomy isomorphism, so it only remains to show that this composition is associative and that every holonomy isomorphism has an inverse.  Associativity is a trivial consequence of the associativity of composition of the maps determining a holonomy isomorphism.  To see that every arrow has an inverse, consider a holonomy isomorphism $(\Psi, \underline{\alpha}, \phi):H\rightarrow H'$, where $H:L_x\rightarrow G$ and $H':L_y\rightarrow G'$. Then $(\Psi^{-1}, \underline{(\Psi \circ \alpha)^{-1}}, \phi^{-1}): H' \rightarrow H$ is a holonomy isomorphism such that $(\Psi, \underline{\alpha}, \phi)\circ (\Psi^{-1}, \underline{(\Psi \circ \alpha)^{-1}}, \phi^{-1})=(\Psi\circ\Psi^{-1},\underline{(\Psi\circ\alpha)^{-1}\bullet(\Psi\circ\alpha)},\phi\circ\phi^{-1})=1_{H'}$ and $(\Psi^{-1}, \underline{(\Psi \circ \alpha)^{-1}}, \phi^{-1})\circ (\Psi, \underline{\alpha}, \phi) = (\Psi^{-1}\circ\Psi,\underline{\alpha\bullet(\Psi^{-1}\circ(\Psi\circ\alpha)^{-1})},\phi^{-1}\circ\phi)= (\Psi^{-1}\circ\Psi,\underline{\alpha\bullet\alpha^{-1}},\phi^{-1}\circ\phi) =1_{H}$, where in both cases the final equalities follow from the fact that for any curve $\alpha$, $\alpha\bullet\alpha^{-1}\in\underline{id_{\alpha(1)}}$, because all holonomy maps agree on thinly equivalent curves.
\end{proof}

There is a natural relationship between $\Hol$ and $\Hol^*$ given by a quotient functors $Q$:
$$
\Hol^* \overset{Q} \longrightarrow \Hol,
$$
where $Q$ acts as the identity on objects, and
$$
Q : (\Psi, \inbrac{\alpha}_{\sim}, \phi) \mapsto (\Psi, \underline{\alpha}, \phi).
$$
Note that $Q$ is indeed well defined, since for all $\alpha, \beta \in \underline{\alpha}$, $\alpha \sim \beta$.

The functor $Q$ clearly preserves empirical content, insofar as that is contained in the information provided by holonomy maps.  Since it is a quotient functor, it is surjective and full.  One can easily confirm, however, that $Q$ is not faithful, by considering a trivial holonomy map $H:L_x\rightarrow G$ mapping all curves $\gamma\in\L_x$ to $id_G$.  Then any closed curves $\alpha,\alpha'\in L_x$ that are not thinly equivalent will yield distinct arrows from $H$ to itself in $\mathbf{Hol}^*$, but these will be mapped to the same arrow in $\mathbf{Hol}$ by $Q$.

The functor $Q$ captures a ``natural'' relationship between $\mathbf{Hol}$ and $\mathbf{Hol}^*$.  Still, one might ask whether there are other relationships of interest between these categories.  In particular, if we are attentive to how we define $\mathbf{Hol}$ and $\mathbf{Hol}^*$ (i.e., we only consider manifolds within some fixed universe of sets), then $Q$ is an epi in $\mathbf{Cat}$, the category of small categories.  One might then wonder if $Q$ splits, i.e., if there is a functor $K:\mathbf{Hol}\rightarrow\mathbf{Hol^{*}}$ such that $Q\circ K =1_{\mathbf{Hol}}$.  If such a functor \emph{did} exist, then it would preserve empirical content, it would be bijective on objects, and it would be faithful, but it would not be full.  If such a functor existed, it would capture another ``natural'' relationship between these categories.  However, no such functor exists.

\begin{prop}\label{split}
$Q$ doesn't split.
\end{prop}

\begin{proof}
If $Q$ split, there would be a functor $K: \mathbf{Hol} \rightarrow \mathbf{Hol}^{*}$ s.t. $Q \circ K = 1_{\mathbf{Hol}}$.  Consider a holonomy map $H$ associated to principal bundle with a flat connection, i.e.,  $H: L_x \rightarrow G$  is such that $H \equiv id_G$. Let $\alpha, \alpha' : [0, 1] \rightarrow M$ be s.t. $\alpha(1) = \alpha'(0) = x$, $\alpha(0) = \alpha'(1)$, and $\alpha$ and the reverse orientation $\alpha'^{-1}$ of $\alpha'$ are not thinly equivalent. Then
$$(id_M, \underline{\alpha'}, id_G ) \circ (id_M, \underline{\alpha}, id_G ) = (id_M, \underline{\alpha \bullet\alpha'}, id_G ) = (id_M, \underline{id_x}, id_G ).$$
 Since this is the identity on $H$ in $\mathbf{Hol}$, $K$ must map it to the identity on $K(H)$ in $\mathbf{Hol}^{*}$, i.e.
 $$K((id_M, \underline{\alpha'}, id_G ) \circ (id_M, \underline{\alpha}, id_G ) ) = (id_M, \inbrac{id_x}_\sim, id_G ).$$
  However, in order to be a functor, $K$ must also satisfy:
\begin{align*}
K((id_M, \underline{\alpha'}, id_G ) \circ (id_M, \underline{\alpha}, id_G ) ) &= K(id_M, \underline{\alpha'}, id_G ) \circ K(id_M, \underline{\alpha}, id_G )\\
&= (id_M, \inbrac{\beta'}_\sim, id_G) \circ (id_M, \inbrac{\beta}_\sim, id_G) \\
&= (id_M, \inbrac{\beta \bullet \beta'}_\sim, id_G)
\end{align*}

Where $\beta \in \underline{\alpha}$ and $\beta' \in \underline{\alpha'}$. These two equations imply that for $\beta = \alpha$ and $\beta' = \alpha'$,
$$
(id_M, \inbrac{id_x}_\sim, id_G ) = (id_M, \inbrac{\alpha \bullet \alpha'}_\sim, id_G),
$$
which in turn implies that $\inbrac{id_x}_\sim = \inbrac{\alpha \bullet \alpha'}_\sim$, which contradicts the assumption that $\alpha$ and $\alpha'^{-1}$ are not thinly equivalent.
\end{proof}

From this we conclude that the only physically interesting relationship between $\mathbf{Hol}$ and $\mathbf{Hol}^*$ is given by $Q$.

\section{Proof of Theorem \ref{thm:PBiso}}

Our proof of Theorem \ref{thm:PBiso}  will depend on the following three lemmas.  In what follows $T_{\Gamma, \gamma}(u)$ denotes the parallel transport via a connection $\Gamma$ on a principal bundle $P$ of a point $u$ along a curve $\gamma: [0, 1] \rightarrow M$ which is such that $\gamma(0) = \pi(u)$. In other words, $T_{\Gamma, \gamma}(u) = \hat{\gamma}_u (1)$.

\begin{lem}\label{lemma1} Let $G \rightarrow P \overset{\pi}{\rightarrow} M$ be a principal bundle and let $\Gamma$ be a principal connection on it. Then for all $x \in M$, $u \in \pi^{-1}[x]$, $\gamma \in L_x$, $g \in G$, and all piece-wise smooth curves $\alpha, \alpha': [0, 1] \rightarrow M$ such that $\alpha(0) = \alpha'(0) = x$ and $\alpha(1) = \alpha'(1)$, the following hold:
\begin{enumerate}
\item[(a)] $T_{\Gamma, \alpha^{-1} \bullet \alpha'}(u) = T_{\Gamma, \alpha^{-1}}(T_{\Gamma, \alpha'}(u))$, where $\alpha^{-1}$ is the reverse orientation of $\alpha$.
\item[(b)] $ T_{\Gamma, \alpha^{-1}}(T_{\Gamma, \alpha'}(u)) = u$ iff $T_{\Gamma, \alpha}(u) = T_{\Gamma, \alpha'}(u)$
\item[(c)] $H_{\Gamma, u}(\gamma) = e_G$, the identity element of $G$, iff $T_{\Gamma, \gamma}(u) = u$
\item[(d)] $T_{\Gamma, \alpha}(ug)  = T_{\Gamma, \alpha}(u)g$
\end{enumerate}
\end{lem}
\begin{proof}
(a) and (b) follow from the fact that every curve $\alpha$ has a unique horizontal lift $\hat{\alpha}_u$ which is such that $\hat{\alpha}_u(0) = u$. (c) follows from the definition of holonomy map. (d) follows from the equivariance of the connection under the right action of $G$ on $P$.
\end{proof}

\begin{lem}\label{lemma2} Let $G \rightarrow P \overset{\pi}{\rightarrow} M$ be a principal bundle and let $\Gamma$ a principal connection on it. Let $\alpha: [0, 1] \rightarrow M$ be a piece-wise smooth curve such that $\alpha(0) = x$ and $\alpha(1) = x'$. Then for all $u \in \pi^{-1}[x']$ and all $\gamma \in L_{x'}$, if $v = T_{\Gamma, \alpha^{-1}}(u) \in \pi^{-1}[x]$, then $$H_{\Gamma, u}(\gamma) = H_{\Gamma, v}(\alpha^{-1} \bullet \gamma \bullet \alpha) $$
\end{lem}
\begin{proof}
Suppose $H_{\Gamma, u}(\gamma) = g \in G$, i.e. that $T_{\Gamma, \gamma}(u) = ug$. Then by Lemma~\ref{lemma1} (b) and (d), $vg = T_{\Gamma, \alpha^{-1}}(u)g = T_{\Gamma, \alpha^{-1}}(ug) = T_{\Gamma, \alpha^{-1}}T_{\Gamma, \gamma}(u) =  T_{\Gamma, \alpha^{-1}}T_{\Gamma, \gamma}T_{\Gamma, \alpha}(v) = T_{\Gamma, \alpha^{-1} \bullet \gamma \bullet \alpha}(v)$. Therefore $H_{\Gamma, v}(\alpha^{-1} \bullet \gamma \bullet \alpha) = g$
\end{proof}

In the following lemma, we make use of the \emph{holonomy sub-bundle} $\Phi_{\Gamma, u} \rightarrow P_{\Gamma, u} \overset{\tilde{\pi}}{\rightarrow} M$ associated with a point $u\in P$ and principal connection $\Gamma$ on a principal bundle $G\rightarrow P\xrightarrow{\pi} M$, as discussed in detail \S II.7 of \citet{Kobayashi+Nomizu}.  This is the bundle consisting of all points of $P$ that may be joined to $u\in P$ by a horizontal curve.  The Reduction Theorem (Theorem II.7.1 of \citet{Kobayashi+Nomizu}) establishes the following about this bundle:
\begin{enumerate}
\item $\Phi_{\Gamma, u} \rightarrow P_{\Gamma, u } \overset{\tilde{\pi}}{\rightarrow} M$ is a reduced sub-bundle of  $G \rightarrow P \overset{\pi}{\rightarrow} M$ with the holonomy group $\Phi_{\Gamma, u}$ as its structure group and with $\tilde{\pi} = \pi_{\upharpoonright  P_{\Gamma, u }}$ (and similarly $P'_{\Gamma', u'}$ is a reduction of $P'$).
\item The connection $\Gamma$ is reducible to a connection $\tilde{\Gamma} = \Gamma_{\upharpoonright \tilde{\pi}}$ on $P_{\Gamma,u}$ (and similarly, $\Gamma'$ reduces to $\tilde{\Gamma}' = \Gamma'_{\upharpoonright \tilde{\pi}'}$).
\end{enumerate}
That $P_{\Gamma, u}$ is a reduced bundle of $P$ means in particular that $\Phi_{\Gamma, u}$ is a Lie subgroup of $G$ and that each element of $P$ may be written (not necessarily uniquely) as $xa$ for some $x\in P _{\Gamma, u}$ and $a\in G$.

\begin{lem} \label{lem:extend} Let $G \rightarrow P \overset{\pi}{\rightarrow} M$ and $G' \rightarrow P' \overset{\pi'}{\rightarrow} M'$ be principal bundles with principal connections $\Gamma$ and $\Gamma'$ respectively, with $M$ and $M'$ connected. Let $\Phi_{\Gamma, u} \rightarrow P_{\Gamma, u} \overset{\tilde{\pi}}{\rightarrow} M$ and $\Phi'_{\Gamma', u'} \rightarrow P'_{\Gamma', u'} \overset{\tilde{\pi'}}{\rightarrow} M'$ be the holonomy sub-bundles of $P$ and $P'$  at $u$ and $u'$, respectively, and $\tilde{\Gamma}$ and $\tilde{\Gamma}'$ be the restrictions of $\Gamma$ and $\Gamma'$ to $P_{\Gamma, u}$ and $P'_{\Gamma', u'}$, respectively. If there is a principal bundle isomorphism $(f, \Psi , \phi_{\upharpoonright \Phi_{\Gamma, u}}) : P_{\Gamma, u} \rightarrow P'_{\Gamma', u'}$ that preserves the connections $\tilde{\Gamma}$ and $\tilde{\Gamma}'$, where $\Psi: M \rightarrow M'$ is a diffeomorphism and $\phi: G \rightarrow G'$ is a Lie group isomorphism, then $(f, \Psi , \phi_{\upharpoonright \Phi_{\Gamma, u}})$ can be extended to a principal bundle isomorphism $(F, \Psi, \phi): P \rightarrow P'$ that preserves $\Gamma$ and $\Gamma'$.
\end{lem}
\begin{proof}
Define $F:P\rightarrow P'$ from $f$ as:
$$F(pg) := f(p)\phi(g) \text{ for } p \in  P_{\Gamma, u}, \text{ } g \in G$$
To prove that $(F, \Psi, \phi)$ is a principal bundle isomorphism, we must show that $F$ is well-defined and a diffeomorphism, and that the following identities hold:
\begin{enumerate}
\item $\pi' \circ F = \Psi \circ \pi$
\item $\pi \circ F^{-1} = \Psi^{-1} \circ \pi'$
\item For all $v \in P$, $g \in G$, $F(vg) = F(v)\phi(g)$
\end{enumerate}
Finally, we must show that $(F,\Psi,\phi)$ preserves $\Gamma$.  We do this by showing that the bundles agree, via the transformation $(F, \Psi, \phi)$, on which curves are horizontal.

To see that $F$ is well-defined, consider any $v \in P$, and suppose there are $x, y \in P_{\Gamma, u}$ and $g,h \in G$ such that $v = xg = yh$. Then $x = yhg^{-1}$, and hence
$$F(xg) = F((yhg^{-1})(g)) = f(yhg^{-1})\phi(g) = f(y)\phi(h)\phi(g^{-1})\phi(g) = f(y)\phi(h) = F(yh)$$

To show that $F$ is also a diffeomorphism, it is sufficient to show that $F$ is bijective and that it is locally a diffeomorphism.  First suppose $F(v) = F(w)$ for some $v, w \in P$. Then by the definition of $F$, $\pi(v) = \pi(w)$, so we may write $v = xg$ and $w = xh$ for the same $x \in P_{\Gamma, u}$. Thus $f(x)\phi(g) = F(v) = F(w) = f(x)\phi(h)$, but since $\phi$ is an isomorphism, this implies that $g = h$ and hence $v = xg = yh = w$.  Thus $F$ is injective.  Now consider any $v' \in P'$.  Write $v' = x'g'$ for some $x' \in P'_{\Gamma', u'}$, $g' \in G'$. Then $F(f^{-1}(x')\phi^{-1}(g')) = x'g' = v'$. Since $f$ and $\phi$ are bijections, $f^{-1}(x')\phi^{-1}(g')$ is a well-defined element of $P$.  So $F$ is bijective.

Finally, let $v \in P$, and let $U \subset M$ be a neighborhood of $\pi(v)$ which is such that a local trivialization of $\pi$ is defined on $U$ and a local trivialization of $\pi'$ is defined on $\Psi[U]$. Then there is a local section $\sigma: U \rightarrow P_{\Gamma, u}$, and $f\circ\sigma\circ\Psi^{-1}$ is a local section of $P'_{\Gamma', u'}$ on $\Psi[U]$. Then for $p \in \pi^{-1}[U]$,
$$F(p) = F(\sigma\circ\pi(p)\theta(p)) = f\circ\sigma\circ\pi(p)\phi\circ\theta(p),$$
where $\theta: \pi^{-1}[U] \rightarrow G$ as $p \mapsto a$, where $a$ is the unique element of $G$ such that $p = \sigma(\pi(p))a$. To see that $\theta$ is smooth, let $\xi: \pi^{-1}[U] \rightarrow U \times G$ be a local trivialization of $P$. Then
$$\theta(p) = ((proj_{R}\circ \xi \circ \sigma \circ \pi)(p))^{-1}(proj_{R} \circ \xi)(p)$$
where $proj_R : U \times G \rightarrow G$ acts as $(z, b) \mapsto b$.
Thus $F_{\upharpoonright \pi^{-1}[U]}$ is the product of compositions of smooth maps, and is hence smooth. The argument for its inverse follows by analogy, once one notes that $F^{-1}(x'g') = f^{-1}(x')\phi^{-1}(g')$.  This completes the argument that $F$ is a diffeomorphism.

We now confirm that the identities 1-3 above hold.  Let $v \in P$.  Then $v = xg$ for some $x \in P_{\Gamma, u}$ and $g \in G$. Since $f$ is an isomorphism and $\pi(v) = \pi(x)$,
$$\pi' \circ F(v) = \pi' (f(x)\phi(g)) = \pi'(f(x)) = \Psi(\pi(x)) = \Psi(\pi(v)). $$  So $\pi'\circ F = \Psi\circ \pi$.  An identical argument establishes that $\pi\circ F^{-1} = \Psi^{-1}\circ\pi'$. Now suppose we have some $v \in P$ and $g \in G$. Then $v = xh$ for some $x \in P_{\Gamma, u}$ and $h \in G$.  It follows that
$$F(vg) = F(xhg) = f(x)\phi(hg) = f(x)\phi(h)\phi(g) = F(v)\phi(g). $$  So $F(vg)=F(v)\phi(g)$, and thus $(F,\Psi,\phi)$ is a principal bundle isomorphism.

It remains to show that $(F,\Psi,\phi)$ preserves $\Gamma$.  Let $\gamma$ be a smooth curve in $M$, $v \in \pi^{-1}(\gamma(0))$, and suppose $v = xg$, $x \in P_{\Gamma, u}$, $g \in G$. Since $\Gamma$ is a principal connection, the lifts of $\gamma$ to $x$ and $v$ are related as
$\hat{\gamma}_v (t) = \hat{\gamma}_x (t)g$. Since $f$ takes $\tilde{\Gamma}$ to $\tilde{\Gamma}'$, we have that
$$F(\hat{\gamma}_v (t)) = F(\hat{\gamma}_x (t)g) = f(\hat{\gamma}_x (t))\phi(g) = \widehat{\Psi \circ \gamma}_{f(x)}(t)\phi(g) = \widehat{\Psi \circ \gamma}_{F(v)}(t).$$  Thus $\Gamma$ and $\Gamma'$  agree on horizontal curves.
\end{proof}

We now turn to the principal result of this section, which we restate here for convenience.
\setcounter{thm}{0}
\begin{thm} Let $G \rightarrow P \overset{\pi}{\rightarrow} M$ and $G' \rightarrow P' \overset{\pi'}{\rightarrow} M'$ be principal bundles with principal connections $\Gamma$ and $\Gamma'$ respectively, and suppose that $M$ and $M'$ are connected.  Suppose there are points $u\in P$ and $u'\in P'$ such that the induced holonomy maps based at $u$ and $u'$ are isomorphic.  Then there is a connection-preserving principal bundle isomorphism between $P$ and $P'$.
\end{thm}
\begin{proof}
 We first show that there is a principal bundle isomorphism $(f, \Psi , \phi) : P_{\Gamma, u} \rightarrow P'_{\Gamma', u'}$ that preserves $\tilde{\Gamma}$, where  $\Phi_{\Gamma, u} \rightarrow P_{\Gamma, u} \overset{\tilde{\pi}}{\rightarrow} M$ and $\Phi'_{\Gamma', u'} \rightarrow P'_{\Gamma', u'} \overset{\tilde{\pi'}}{\rightarrow} M'$ are the holonomy sub-bundles of $P$ and $P'$  at $u$ and $u'$, respectively, and $\tilde{\Gamma}$ and $\tilde{\Gamma}'$ are the restrictions of $\Gamma$ and $\Gamma'$ and $P_{\Gamma, u}$ and $P'_{\Gamma', u'}$, respectively.  We then invoke Lemma \ref{lem:extend} to extend $(f, \Psi, \phi)$ to a principal bundle isomorphism  $(F, \Psi, \phi): P \rightarrow P'$ that preserves $\Gamma$.

First, since $H_{\Gamma,u}$ and $H'_{\Gamma',u'}$, the holonomy maps induced by $\Gamma$ and $\Gamma'$ and based at $u$ and $u'$, respectively, are isomorphic by assumption, there must be some holonomy isomorphism $(\Psi,\underline{\alpha},\phi):H_{\Gamma,u}\rightarrow H'_{\Gamma',u'}$.  Let $z := T_{\Gamma, \alpha^{-1}}(u) \in \pi^{-1}(\alpha(0))$, where $\alpha\in\underline{\alpha}$.  (Note that $z \in P_{\Gamma, u}$, and moreover $P_{\Gamma, u} = P_{\Gamma, z}$, i.e., every element of $P_{\Gamma, u}$ can be connected to $z$ via some piece-wise smooth, horizontal curve). Define $f: P_{\Gamma, u} \rightarrow P'_{\Gamma', u'}$ as follows:
\begin{enumerate}
\item[(i)] $f(z) := u'$
\item[(ii)] For any $v \in P_{\Gamma, u}$, pick some piece-wise smooth curve $\beta_{v} \in C_{M, \pi(z)}$ (where $C_{M, \pi(z)}$ denotes the set of piece-wise smooth space-time curves $\gamma : [0, 1] \rightarrow M$ such that $\gamma(0) =\pi(z) = \alpha(0)$) such that $v = T_{\tilde{\Gamma}, \beta_{v}}(z)$, the parallel transport in $P_{\Gamma,u}$ of $z$ along $\beta_v$ according to the connection $\tilde{\Gamma}$. Then set $f(v) := T'_{\tilde{\Gamma}', \Psi \circ \beta_{v}}(u')$, the parallel transport in $\tilde{\pi}$ of $u'$ along $\Psi \circ \beta_v$.
\end{enumerate}

We claim that the triple $(f, \Psi, \phi)$ realizes the desired principal bundle isomorphism. To prove this, we must show that $f$ is well-defined, a diffeomorphism, and that the following identities hold:
\begin{enumerate}
\item $\tilde{\pi}' \circ f = \Psi \circ \tilde{\pi}$
\item $\tilde{\pi} \circ f^{-1} = \Psi^{-1} \circ \tilde{\pi}'$
\item For all $v \in P_{\Gamma, u}$, $g \in \Phi_{\Gamma, u}$, $f(vg) = f(v)\phi(g)$
\end{enumerate}
Finally, we must show that $(f, \Psi, \phi)$ preserves the reduced connection $\tilde{\Gamma}$.

We begin by showing that $f$ is well-defined.  Consider any point $v \in P_{\Gamma, u}$. Suppose the curves $\beta$ and $\beta' \in C_{M, \pi(z)}$ are such that $T_{\tilde{\Gamma}, \beta}(z) = T_{\tilde{\Gamma}, \beta'}(z) = v$. We want to show that $T'_{\tilde{\Gamma}',\Psi \circ \beta}(u') = T'_{\tilde{\Gamma}',\Psi \circ \beta'}(u')$. Let $\beta^{-1}$ denote the reverse orientation of $\beta$, and $e_{G}$ the identity element of G (and hence of $\Phi_{\Gamma, z}$ and $\Phi_{\Gamma, u}$). By Lemma~\ref{lemma1} (a) and (b), $T_{\tilde{\Gamma}, \beta^{-1} \bullet \beta'}(z) = T_{\tilde{\Gamma},  \beta^{-1}}(T_{\tilde{\Gamma},  \beta'}(z)) = T_{\tilde{\Gamma}, \beta^{-1}}(v) = z$. Thus by Lemma~\ref{lemma1} (c), $H_{\Gamma, z}(\beta^{-1} \bullet \beta') = e_{G}$. Since $\phi$ is a Lie group isomorphism, we  also know that $\phi(e_{G} ) = e_{G'}$. By Lemma~\ref{lemma2}, then, we know that  $e_{G} = H_{\Gamma, z}(\beta^{-1} \bullet \beta') = H_{\Gamma, u}(\alpha \bullet \beta^{-1} \bullet \beta' \bullet \alpha^{-1}) =  H_{\Gamma, u}(\bar{\alpha}^{-1}(\beta^{-1} \bullet \beta'))$, where $\bar{\alpha}$ is as in Def. \ref{iso}. Since $(\Psi, \underline{\alpha}, \phi)$ is a holonomy isomorphism, we know that $e_{G'} = \phi\circ H_{\Gamma, u}(\bar{\alpha}^{-1}(\beta^{-1} \bullet \beta')) = (H_{\Gamma', u'} \circ \psi \circ \bar{\alpha})(\bar{\alpha}^{-1}(\beta^{-1} \bullet \beta')) = H_{\Gamma', u'} (\Psi \circ (\beta^{-1} \bullet \beta'))$. This tells us that $u' = T'_{\tilde{\Gamma}',\Psi \circ (\beta^{-1} \bullet \beta')}(u') = T'_{\tilde{\Gamma}',\Psi \circ \beta^{-1}}(T'_{\tilde{\Gamma}',\Psi \circ \beta'}(u'))$. By Lemma~\ref{lemma1} (b), this implies that $T'_{\tilde{\Gamma}',\Psi \circ \beta}(u') = T'_{\tilde{\Gamma}',\Psi \circ \beta'}(u')$.  So $f$ is well-defined.

We now show that $f$ is bijective.  (Later we will also show that $f$ and $f'$ are smooth, completing the proof that $f$ is a diffeomorphism.)  Let $v, w \in P_{\Gamma, u}$, and suppose $f(v) = f(w)$. We want to show that $v = w$. Since $f(v) = f(w)$, we know that $T'_{\tilde{\Gamma}', \Psi \circ \beta_v}(u') = f(v) = f(w) = T'_{\tilde{\Gamma}', \Psi \circ \beta_w}(u')$.  By Lemma~\ref{lemma1} (a) and (b) and the fact that $\Psi$ is a diffeomorphism, we get that $u' =  T'_{\tilde{\Gamma}', (\Psi \circ \beta_{v})^{-1}}(T'_{\tilde{\Gamma}', \Psi \circ \beta_w }(u'))  = T'_{\tilde{\Gamma}', (\Psi \circ \beta_{v})^{-1} \bullet (\Psi \circ \beta_w )}(u') = T'_{\tilde{\Gamma}', \Psi \circ (\beta^{-1}_{v} \bullet \beta_w )}(u') $. Thus by Lemma~\ref{lemma1} (c) we get that $H_{\Gamma', u'}(\Psi \circ (\beta^{-1}_{v} \bullet \beta_w )) = e_{G'}$. Since $(\Psi, \underline{\alpha}, \phi)$ is a holonomy isomorphism, this implies that $\phi(H_{\Gamma, u}(\bar{\alpha}^{-1}(\beta^{-1}_{v} \bullet \beta_w ))) = e_{G'}$, which, since $\phi$ is a Lie group isomorphism, implies that $H_{\Gamma, u}(\bar{\alpha}^{-1}(\beta^{-1}_{v} \bullet \beta_w )) = e_G$.  By Lemma~\ref{lemma2}, then, $H_{\Gamma, z}(\beta^{-1}_{v} \bullet \beta_w ) = e_G$. Thus by Lemma~\ref{lemma1} (c), $v = T_{\tilde{\Gamma}, \beta_v }(z) =  T_{\tilde{\Gamma}, \beta_w }(z) = w$.  So $f$ is injective.  Now let $w' \in P'_{\Gamma',u'}$, and let the curve $\beta' \in C_{M', \pi'(u')}$ be such that $T'_{\tilde{\Gamma}',\beta'}(u') = w'$.  Then there is a unique $v \in P_{\Gamma, u}$ such that $v = T_{\tilde{\Gamma},\Psi^{-1} \circ \beta'}(z)$. Then $f(v) = T'_{\tilde{\Gamma'}, \Psi \circ \alpha_{v}}(u')= T'_{\tilde{\Gamma'}, \Psi\circ(\Psi^{-1} \circ \beta')}(u') = T'_{\tilde{\Gamma'}, \beta'}(u') = w'$. (The second equality follows from fact that f is well-defined.)  It follows that $f$ is bijective.

We will now establish identities 1-3.  Let $v \in P_{\Gamma, u}$.  Then $$\tilde{\pi}'(f(v)) = \tilde{\pi}'(T'_{\tilde{\Gamma}',\Psi \circ \beta_{v}}(u')) = (\Psi \circ \beta_{v})(1) = \Psi(\beta_{v}(1)) = \Psi(\tilde{\pi}(v)).$$
So $\tilde{\pi}'\circ f = \Psi\circ\tilde{\pi}$.  By identical reasoning, $\tilde{\pi}\circ f^{-1} = \Psi^{-1}\circ\tilde{\pi}'$. Finally, let $v \in P_{\Gamma, u}$ and $g \in \Phi_{\Gamma, u}$. First note that by Lemma~\ref{lemma1} (d) and the well-definedness of $f$, we can assume without loss of generality that $\beta_{vg} = \beta_v \bullet \beta_{zg}$. By Lemma~\ref{lemma1} (a), $f(vg) = T'_{\tilde{\Gamma}', \Psi \circ (\beta_v \bullet \beta_{zg})}(u') =  T'_{\tilde{\Gamma}', \Psi \circ \beta_v}(T'_{\tilde{\Gamma}', \Psi \circ  \beta_{zg}}(u'))$.  By the definition of holonomy isomorphism, $T_{\tilde{\Gamma}', \Psi \circ \beta_{zg}}(u') = u'H_{\Gamma', u'}(\Psi \circ\beta_{zg}) = u'H_{\Gamma', u'}(\Psi \circ \bar{\alpha} \circ (\alpha \bullet \beta_{zg} \bullet \alpha^{-1})) = u'\phi(H_{\Gamma, u}(\alpha \bullet \beta_{zg} \bullet \alpha^{-1})) = u'\phi(H_{\Gamma, z}(\beta_{zg})) = u'\phi(g)$. Plugging this equality into the last one, and using Lemma~\ref{lemma1} (d), we get: $f(vg) = T'_{\tilde{\Gamma}', \Psi \circ \beta_v}(u'\phi(g)) = T'_{\tilde{\Gamma}', \Psi \circ \beta_v}(u')\phi(g) = f(v)\phi(g)$.

Next we show that $f$ preserves $\Gamma$. It suffices to show that for all piece-wise smooth curves $\gamma: [0, 1] \rightarrow M$  and all $w \in \pi^{-1}(\gamma(0))$, $f(T_{\tilde{\Gamma}, \gamma}(w)) = T'_{\tilde{\Gamma}, \Psi \circ \gamma}f(w)$. But this follows easily from the definition of $f$:  $f(T_{\tilde{\Gamma}, \gamma}(w)) = T'_{\tilde{\Gamma}, \Psi \circ \beta_{T_{\tilde{\Gamma}, \gamma}(w)}}(u') = T'_{\tilde{\Gamma}, \Psi \circ (\gamma \bullet \beta_{w})}(u') = T'_{\tilde{\Gamma}, \Psi \circ \gamma}(T'_{\tilde{\Gamma}, \Psi \circ  \beta_{w}}(u')) =  T'_{\tilde{\Gamma}, \Psi \circ \gamma}(f(w))$.

To complete the proof, we have only to show that $f$ and $f^{-1}$ are smooth.  Then $f$ will be a diffeomorphism, and $(f,\Psi,\phi)$ will be a principal bundle isomorphism that preserves $\Gamma$. Let $v \in P_{\Gamma, u}$ and let $V \subseteq M$ an open neighborhood of $x = \tilde{\pi}(v)$ on which a local trivialization of $P_{\Gamma, u}$ is defined. Let $V'$ be a neighborhood of $\Psi(x)$ on which a local trivialization of $P'_{\Gamma', u'}$ is defined. Let $g$ be a metric on $M$, $g' = \Psi_{*}(g)$. Let $U$ be an open subset of $V \cap \Psi^{-1}[V']$ (containing $x$) on which the exponential map $exp_{x}$ is a diffeomorphism from a subset $U_x \subseteq T_x M$ onto $U$.

By definition, $exp_{x}(\xi) = \gamma_{\xi}(1)$, where $\gamma_{\xi}$ is a $g$-geodesic in $M$ such that $\left(\frac{d}{dt}\gamma_{\xi}\right)_{t=0} = \xi$. We may also ``lift" $exp_x$ to $v$ by defining $\widehat{exp}_v : U_x \rightarrow P_{\Gamma, u}$, where $\xi \mapsto (\hat{\gamma}_{\xi})_v (1)$. Similarly we may define $exp_{\Psi(x)}: U'_{\Psi(x)} \rightarrow P'_{\Gamma', u'}$ on $M'$ using $g'$, in which case $U'_{\Psi(x)} = \Psi_{*}[U_x]$, and for any $\xi'\in U'_{\Psi(x)}$,
$$exp_{\Psi(x)}(\xi') = \gamma_{\xi'}(1) = \Psi \circ \gamma_{\Psi^{*}(\xi')}(1) = \Psi \circ exp_{x}(\Psi^{*}(\xi'))$$
since $g' = \Psi_{*}(g)$.  (Recall that since $\Psi$ is a diffeomorphism, we may define the pullback of vectors as $\Psi^* = (\Psi^{-1})^*$.)  We also get that
\begin{align*}
\widehat{exp}_{f(v)}(\xi') &= (\hat{\gamma}_{\xi'})_{f(v)}(1) = (\widehat{\Psi \circ \gamma_{\Psi^{*}(\xi')}})_{f(v)} \\
&= T'_{\Gamma', \Psi \circ ( \gamma_{\Psi^{*}(\xi')} \bullet \beta_{v})}(u') \\
&= f(T_{\Gamma, \gamma_{\Psi^{*}(\xi')} \bullet \beta_{v}}(u)) \\
&= f \circ \widehat{exp}_{v}(\Psi^{*}(\xi')).
\end{align*}

Now define a smooth local section $\sigma: U \rightarrow P_{\Gamma, u}$ as $\sigma = \widehat{exp}_{v} \circ exp_{x}^{-1}$. Then
$$\sigma' = f \circ \sigma  \circ \Psi^{-1} = f \circ \widehat{exp}_{v} \circ exp_{x}^{-1} \circ \Psi^{-1} = \widehat{exp}_{f(v)} \circ \Psi_{*} \circ exp_{x}^{-1} \circ \Psi^{-1}$$

is a smooth local section of $P'_{\Gamma', u'}$. Now let $\eta: U \times \Phi_{\Gamma, u} \rightarrow \tilde{\pi}^{-1}[U]$ be a local trivialization of $P_{\Gamma, u}$ such that $\eta^{-1}[\sigma[U]] = U \times \{e_G \}$, and let $\eta': \Psi[U] \times \Phi'_{\Gamma', u'} \rightarrow \tilde{\pi}'^{-1}[\Psi[U]]$ be a local trivialization of $P'_{\Gamma', u'}$ such that $\eta^{-1}[\sigma'[\Psi[U]]] = \Psi[U] \times \{e_{G'} \}$. Then we can write $f$ locally as
$$f_{\upharpoonright U} = \eta' \circ (\Psi \times \phi) \circ \eta^{-1} $$
since for all $w \in \tilde{\pi}^{-1}[U]$, we can write $w = yg$ for some $y \in \sigma[U]$. Then
\begin{align*}
 \eta' \circ (\Psi \circ \phi) \circ \eta^{-1}(w) &= \eta' \circ (\Psi \circ \phi) (\tilde{\pi}(w), g) \\
 &= \eta' (\Psi \circ \tilde{\pi}(w), \phi(g)) \\
 &=  \eta' (\Psi \circ \tilde{\pi}(w), e_{G'})\phi(g) \\
 &= \sigma'(\Psi \circ \tilde{\pi}(w))\phi(g) \\
 &= f \circ \sigma \circ \Psi^{-1}(\Psi \circ \tilde{\pi}(w))\phi(g) \\
 &= f \circ \sigma \circ \tilde{\pi}(w) \phi(g) \\
 &= f(y)\phi(g) = f(w).
\end{align*}

Since $v$ was arbitrary, $f$ is smooth everywhere. An analogous procedure can be performed for $f^{-1}$.\end{proof}

\section{Proof of Theorem \ref{thm:cat}}

We now prove the main result.  Again, we restate it first for convenience.

\begin{thm}$\mathbf{Hol}$ and $\mathbf{PC}$ are equivalent as categories, with an equivalence that preserves empirical content in the sense of preserving holonomy data.
\end{thm}
\begin{proof}
Let $C: \mathbf{Hol} \rightarrow \mathbf{PC}$ be a functor that takes holonomy maps $H:L_x \rightarrow G$ on a manifold $M$ to a principal bundle  $G \rightarrow P \overset{\pi}{\rightarrow} M$ and principal connection $\Gamma$ given by the Barrett reconstruction theorem---i.e., to a bundle and connection $(G\rightarrow P\xrightarrow{\pi} M,\Gamma)$ such that there exists a point $u\in \pi^{-1}[x]$ satisfying $H_{\Gamma,u}=H$---and takes a holonomy isomorphism $(\Phi, \underline{\alpha}, \phi)$ to the principal bundle isomorphism $(F, \Psi, \phi): C(H_{\Gamma,u}) \rightarrow C(H'_{\Gamma',u'})$ given in the proof of Theorem~\ref{thm:PBiso}.  First, note that $C$ clearly preserves holonomy data, and thus preserves empirical content in the required sense.  We will first show that $C$ is indeed a functor, and then show that $C$ is one half of an equivalence, by showing it is full, faithful, and essentially surjective.

First, it is clear from the definition of $F$ that $C((\Psi, \underline{\alpha}, \phi): H \rightarrow H') = (F, \Psi, \phi): C(H) \rightarrow C(H')$. It remains to show that $C(id_H) = id_{C(H)}$ and that $C(g \circ f) = C(g) \circ C(f)$ for any arrows $f: H \rightarrow H'$ and $g: H' \rightarrow H''$ of $\mathbf{Hol}$.  So let $H$ be an arrow of $\mathbf{Hol}$, suppose $C(H) = (G \rightarrow P \overset{\pi}{\rightarrow} M, \Gamma)$, and suppose $u\in\pi^{-1}[x]$ is such that $H_{\Gamma,u}=H$. Then $C(id_H) = C((id_M, \underline{id_{\pi(u)}}, id_G )) = (id_P, id_M, id_G) = id_{C(H)}$.  Thus identities are preserved.  Now let $(\Psi, \underline{\alpha}, \phi): H \rightarrow H'$ and $(\Psi', \underline{\alpha}', \phi'): H' \rightarrow H''$ be isomorphisms of holonomy maps $H:L_x \rightarrow G$, $H':L_{x'}\rightarrow G'$ and $H'':L_{x''}\rightarrow G''$.  Let $(P, \Gamma)$, $(P', \Gamma')$, and $(P'', \Gamma'')$ be the corresponding principal bundles and connections in the Barrett construction, and let $u\in \pi^{-1}[x]$, $u'\in\pi'^{-1}[x']$, and $u''\in\pi''^{-1}[x'']$ be such that $H = H_{\Gamma, u}$, $H = H_{\Gamma', u'}$, and $H = H_{\Gamma'', u''}$, respectively.  Then
\begin{align*}
&C((\Psi, \alpha, \phi) \circ (\Psi', \alpha', \phi')) : H \rightarrow H'' \\
&= C(\Psi' \circ \Psi, \alpha \bullet (\Psi^{-1} \circ \alpha' ), \phi' \circ \phi): H \rightarrow H''\\
&= (F'', \Psi' \circ \Psi, \phi' \circ \phi): C(H) \rightarrow C(H'')
\end{align*}
Where for $v \in P$, if $v = xg$ for $x \in P_{\Gamma, u}$, $g \in G$, then
\begin{align*}
F''(v) &= T_{\Gamma'', \Psi^{-1} \circ \Psi (\alpha \bullet (\Psi^{-1} \circ \alpha'))}(u'')(\phi' \circ \phi)(g) \\
&= F' (f(x)\phi(g)) = (f' \circ f)(x)(\phi' \circ \phi)(g)\\
&= F' \circ F(v)
\end{align*}

We now show that $C$ is full, faithful, and essentially surjective. Let $H:L_x \rightarrow G$ and $H':L_{x'}\rightarrow G'$ be objects of $\mathbf{Hol}$, and suppose $C(H) = (G \rightarrow P \overset{\pi}{\rightarrow} M, \Gamma)$, $C(H') = (G' \rightarrow P' \overset{\pi'}{\rightarrow} M', \Gamma')$, where $u \in \pi^{-1}[x]$ and $u' \in \pi'^{-1}[x']$ are such that $H = H_{\Gamma, u}$ and $H' = H_{\Gamma', u'}$. Suppose there is an isomorphism $(F', \Psi, \phi): (P, \Gamma) \rightarrow (P', \Gamma')$ of the principal bundles and connections. Let $\alpha$ be a piece-wise smooth curve in $M$ such that $\alpha(0) = \Psi^{-1}(x')$, $\alpha(1) = x$, and $\hat{\alpha}_{F'^{-1}(u')}(1) = u$. We claim that $C((\Psi, \underline{\alpha}, \phi)) = (F', \Psi, \phi)$. For let $(F, \Psi, \phi)$ be the isomorphism corresponding to $(\Psi, \underline{\alpha}, \phi)$ given in Theorem~\ref{thm:PBiso}, and suppose $v \in P$ is such that $v = yg$ for some $y \in P_{\Gamma, u}$ and $g \in G$. Then $F(v) = T'_{\Gamma', \Psi \circ \beta_y}(u')\phi(g)$ for some piece-wise smooth curve $\beta_y : [0, 1] \rightarrow M$ such that $\beta_y (0) = \Psi^{-1}(x')$ and $\widehat{\beta_y}_z (1) = y$, where $\widehat{\beta_y}_z$ is the lift of $\beta_y$ through $z = T_{\Gamma, \alpha^{-1}}(u) = F'^{-1}(u')$.
Thus $T'_{\Gamma', \Psi \circ \beta_y}(u') = \widehat{\Psi \circ \beta_y}_{u'}(1) = F'(\beta_y (1)) = F'(y)$ by the definition of principal connection. Thus $F(v) = F'(y)\phi(g) = F'(v)$, so $C(\Psi, \underline{\alpha}, \phi) = (F', \Psi, \phi)$. So $C$ is full.

Now suppose there are two holonomy isomorphisms $(\Psi, \underline{\alpha}, \phi)$ and $(\Psi', \underline{\alpha}', \phi'): H \rightarrow H'$ which are such that $C(\Psi, \underline{\alpha}, \phi) = C(\Psi', \underline{\alpha}', \phi') = (F, \Psi'', \phi'')$. Then by the definition of $C$ on arrows, $\Psi = \Psi' = \Psi''$ and $\phi = \phi' = \phi''$. Thus for all $\gamma \in L_x$,
$$H'(\Psi\circ(\alpha^{-1} \bullet \gamma \bullet \alpha)) = \phi\circ H (\gamma) = \phi'\circ H(\gamma) = H'(\Psi\circ(\alpha'^{-1} \bullet \gamma \bullet \alpha')) $$
 Thus $\underline{\alpha}=\underline{\alpha'}$, and so $(\Psi,\underline{\alpha},\phi)=(\Psi',\underline{\alpha'},\phi')$ and $C$ is faithful.  Finally, let  $G \rightarrow P \overset{\pi}{\rightarrow} M$ be a principal bundle with connection $\Gamma$, $(P, \Gamma) \in \mathbf{PC}$. Then $C(H_{\Gamma, u}) = (P, \Gamma)$ for some $u \in P$.  So $C$ is essentially surjective.
\end{proof}

\section{Discussion}\label{sec:discussion}

We have now proved the main results of the paper.  In particular, Theorem \ref{thm:cat} establishes that on at least one construal of the category of holonomy models, $\mathbf{Hol}$ and $\mathbf{PC}$ are equivalent.  This captures one sense in which one might think that no structure is lost in moving between principal bundle and loop formulations of Yang-Mills theory; one might also take it to capture a sense in which these formalisms are equivalent, by virtue of having the capacity to represent just the same physics.

We conclude with two comments.  The first is to note a relationship to the \citeyear{Schreiber+Waldorf} result of \citet{Schreiber+Waldorf}, who showed that category we call $\mathbf{PC}$ is equivalent to a category $\mathbf{Trans}$ of parallel transport maps and suitably defined arrows.  It follows that $\mathbf{Hol}$ is also equivalent to $\mathbf{Trans}$, and that $\mathbf{Hol}^*$ is not equivalent to any of these other categories, at least in a physically interesting way.  Thus the results presented might thus be seen provide a broader picture of the categorical relationships between three important ways of characterizing models of Yang-Mills theory.  The situation is summarized in the following Figure \ref{fig1}.

\begin{figure}[!h]
\centering
\includegraphics[width=.7\linewidth]{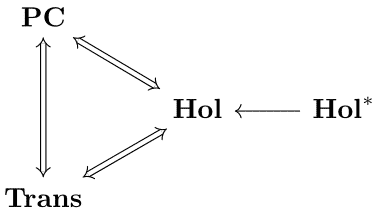}
\caption{A representation of the categorical relationships between the category $\mathbf{PC}$ of principal bundles with principal connections, the category $\mathbf{Trans}$ of parallel transport maps, and two possible choices of categories of holonomy maps. Double arrows denote categorical equivalence, while single arrows denote the existence of a quotient functor that does not split. }
\label{fig1}
\end{figure}

The second comment concerns how to understand the relationships just sketched in terms of the Baez-Dolan-Bartels classification for forgetful functors.  Of course, the functors realizing equivalences forget nothing, and so there is an important sense in which $\mathbf{PC}$, $\mathbf{Hol}$, and $\mathbf{Trans}$ all encode precisely the same information---i.e., they can be related by functors that forget nothing.  As we have noted, however, the functor $Q:\mathbf{Hol}^*\rightarrow\mathbf{Hol}$ is not essentially invertible: it is full and surjective, but not faithful.  Thus it forgets (only) ``stuff''.  (What stuff?  It is not clear that a clean answer is available, but our intuition is that we are forgetting unnecessary information about the base point.)  Likewise, if $F:\mathbf{Hol}\rightarrow\mathbf{PC}$ is the functor that realizes the equivalence in Theorem \ref{thm:cat}, then $F\circ Q:\mathbf{Hol^{*}}\rightarrow\mathbf{PC}$ also forgets only stuff.

It is in this context that Prop. \ref{split} becomes particularly interesting.  There we show that $Q$ does not split.  If it \emph{did} split, then there would be a functor $K:\mathbf{Hol}\rightarrow\mathbf{Hol}^*$ that would be faithful and essentially surjective, but not full---i.e., it would forget (only) ``structure''.  It would follow that there would be a functor $K\circ F^{-1}:\mathbf{PC}\rightarrow \mathbf{Hol^{*}}$ that also forgot only structure.  One might then argue that there is a sense in which the holonomy formalism has less structure than the principal bundle formalism, provided one could argue that $\mathbf{Hol^{*}}$ is otherwise preferable to $\mathbf{Hol}$.

It is this argument that is blocked by Prop. \ref{split}, removing the worry that choosing $\mathbf{Hol}$ over $\mathbf{Hol^{*}}$ somehow ``adds'' structure (in the Baez-Dolan-Bartel sense) in a way the undermines the significance of Theorem \ref{thm:cat}.  In other words, if we chose to work with $\mathbf{Hol^{*}}$ instead of $\mathbf{Hol}$, it would seem that the holonomy formalism would have more ``stuff'' than the principal bundle formalism, and no less structure.  This helps clarify what is at stake in choosing between $\mathbf{Hol}$ and $\mathbf{Hol^{*}}$; it also gives some reason to doubt that either choice of category will help someone who believes the holonomy formalism is somehow more parsimonious.

That said, there is some sense in which these final considerations are beside the point.  One might have thought that the question of real interest was whether or not the principal bundle formalism allows us to describe physical situations that, by the lights of that theory, properly construed, are somehow physically inequivalent---say because they require us to make a choice between different, inequivalent bundle structures---but which nonetheless correspond to the same holonomy data. One might then think that the principal bundle formalism has some sort of ``excess structure'', such that we would need to posit a new form of isomorphism between principal bundles with connections, analogous to the gauge transformations one introduces in classical electromagnetism, to remove that structure.\footnote{For more on this way of thinking about this issue, see \citet{WeatherallTE} and \citet{WeatherallGauge}.}  But if that were the worry, then Theorem \ref{thm:PBiso} substantially settles the issue, since it establishes that given an equivalence class of holonomy models, in the sense of $\mathbf{Hol}$ \emph{or} $\mathbf{\Hol}^*$, then there is a unique principal bundle and principal connection with the appropriate structure group that gives rise to those holonomy models.

\section*{Acknowledgments}
This material is based upon work supported by the National Science Foundation under Grant No. 1331126.  Thank you to David Malament and two anonymous referees for helpful comments on a previous draft and to audiences at the Southern California Philosophy of Physics Group, the 2014 Irvine-Pittsburgh-Princeton Conference on the Mathematical and Conceptual Foundations of Physics, and the Munich Center for Mathematical Philosophy for comments and discussion.

\end{document}